\documentclass[english,pra,a4paper,aps,showpacs,twocolumn,floats]{revtex4}
\usepackage[latin9]{inputenc}
\usepackage{amsmath}
\usepackage{amsthm}
\usepackage{amssymb}
\usepackage{graphicx}
\usepackage{xcolor}

\makeatletter
\@ifundefined{textcolor}{}
{%
 \definecolor{BLACK}{gray}{0}
 \definecolor{WHITE}{gray}{1}
 \definecolor{RED}{rgb}{1,0,0}
 \definecolor{GREEN}{rgb}{0,1,0}
 \definecolor{BLUE}{rgb}{0,0,1}
 \definecolor{CYAN}{cmyk}{1,0,0,0}
 \definecolor{MAGENTA}{cmyk}{0,1,0,0}
 \definecolor{YELLOW}{cmyk}{0,0,1,0}
}
\theoremstyle{plain}
\newtheorem{thm}{\protect\theoremname}
 \theoremstyle{definition}
 \newtheorem{defn}[thm]{\protect\definitionname}
 \theoremstyle{plain}
 \newtheorem{cor}[thm]{\protect\corollaryname}

\usepackage{MnSymbol}
\usepackage{slashbox}
\usepackage{hyperref}

\makeatother

\usepackage{babel}
 \providecommand{\corollaryname}{Corollary}
 \providecommand{\definitionname}{Definition}
\providecommand{\theoremname}{Theorem}

\begin{document}

\title{Probing finite coarse-grained virtual Feynman histories with sequential weak values}

\author{Danko Georgiev}

\affiliation{Institute for Advanced Study, Varna, Bulgaria}

\email{danko.georgiev@mail.bg}

\author{Eliahu Cohen}

\affiliation{Physics Department, Centre for Research in Photonics, University of Ottawa,
Advanced Research Complex, 25 Templeton, Ottawa ON Canada, K1N 6N5}

\email{eli17c@gmail.com}


\date{May 3, 2018}
\begin{abstract}
Feynman's sum-over-histories formulation of quantum mechanics has been considered a useful calculational tool in which virtual Feynman histories entering into a coherent quantum superposition cannot be individually measured. Here we show that sequential weak values, inferred by consecutive weak measurements of projectors, allow direct experimental probing of individual virtual Feynman histories thereby revealing the exact nature of quantum interference of coherently superposed histories. Because the total sum of sequential weak values of multi-time projection operators for a complete set of orthogonal quantum histories is unity, complete sets of weak values could be interpreted in agreement with the standard quantum mechanical picture. We also elucidate the relationship between sequential weak values of quantum histories with different coarse-graining in time and establish the incompatibility of weak values for non-orthogonal quantum histories in history Hilbert space. Bridging theory and experiment, the presented results may enhance our understanding of both weak values and quantum histories.
\end{abstract}

\maketitle
In this work, we revisit the important yet controversial concept of
quantum weak values and elucidate the relationship between Aharonov's
two-state vector formalism and Feynman's sum-over-histories.
%%%
This interesting relationship resonates with past works which studied non-demolition and continuous quantum measurements \cite{Braginsky1992,Braginsky1996,Mensky1993,Mensky2000}, while connecting them with path integration \cite{Mensky1979}. Recently, the above relationship was further analyzed and strengthened by different researchers \cite{Duprey2017,Matzkin2012,Matzkin2015,Sokolovski2016a,Sokolovski2016b,Sokolovski2017a}, but here we focus on the notion of sequential weak values as a pivotal issue, which has not been mentioned before in the above literature.
%%%
In particular, we show that sequential weak values are able to probe directly the
quantum probability amplitudes along individual virtual Feynman histories
thereby possibly supporting their physical meaningfulness. Conversely, we utilize
the mathematical constraints behind Feynman summation in order to provide rules for consistent interpretation of experimentally measured
weak values.
\section{Preliminaries}
To begin with, we succinctly describe a finite coarse-grained Feynman's
sum-over-histories procedure applicable to any experiment performed with a finite precision.
\begin{defn}
(Quantum history) Quantum histories from an initial time $t_{i}$
to a final time $t_{f}$ are constructed at $k+2$ different times
$t_{i},t_{1},t_{2},\ldots,t_{k},t_{f}$ with the use of complete sets
of projection operators $\{\mathcal{\hat{P}}_{i,1},\hat{\mathcal{P}}_{i,2},\ldots,\mathcal{\hat{P}}_{i,n}\}$,
$\{\mathcal{\hat{P}}_{1,1},\hat{\mathcal{P}}_{1,2},\ldots,\mathcal{\hat{P}}_{1,n}\}$,
$\{\mathcal{\hat{P}}_{2,1},\hat{\mathcal{P}}_{2,2},\ldots,\mathcal{\hat{P}}_{2,n}\}$,
$\ldots$, $\{\mathcal{\hat{P}}_{k,1},\hat{\mathcal{P}}_{k,2},\ldots,\mathcal{\hat{P}}_{k,n}\}$,
$\{\mathcal{\hat{P}}_{f,1},\hat{\mathcal{P}}_{f,2},\ldots,\mathcal{\hat{P}}_{f,n}\}$
which at each single time span the $n$-dimensional Hilbert space
of the system $\sum_{n}\mathcal{\hat{P}}_{i,n}=\hat{I}$, $\sum_{n}\mathcal{\hat{P}}_{1,n}=\hat{I}$,
$\sum_{n}\mathcal{\hat{P}}_{2,n}=\hat{I}$, $\ldots$, $\sum_{n}\mathcal{\hat{P}}_{k,n}=\hat{I}$,
$\sum_{n}\mathcal{\hat{P}}_{f,n}=\hat{I}$. Using the symbol $\odot$
for tensor products at different times,
we can write each \emph{quantum history} as a projection
operator $\hat{\mathcal{Q}}_{j}=\mathcal{\hat{P}}_{f}\odot\mathcal{\hat{P}}_{k}\odot\ldots\odot\mathcal{\hat{P}}_{2}\odot\mathcal{\hat{P}}_{1}\odot\mathcal{\hat{P}}_{i}$
in \emph{history Hilbert space} $\breve{\mathcal{H}}=\mathcal{H}_{f}\odot\mathcal{H}_{k}\odot\ldots\odot\mathcal{H}_{2}\odot\mathcal{H}_{1}\odot\mathcal{H}_{i}$,
where $\mathcal{H}_{k}$ is a copy of the standard Hilbert space of
the physical system at time $t_{k}$ \cite{Gell-Mann1990,Gell-Mann1993,Griffiths1984,Griffiths1993,Griffiths2003,Halliwell1995,Hartle1993}. By construction there are $n^{k+2}$~orthogonal quantum histories ($\textrm{Tr}[\mathcal{\hat{Q}}_{j}\mathcal{\hat{Q}}_{j'}]=0$ for $j\neq j'$) that span the history Hilbert space $\breve{\mathcal{H}}$.
\end{defn}

\begin{defn}
(Chain operator) To each quantum history $\hat{\mathcal{Q}}_{j}=\mathcal{\hat{P}}_{f}\odot\mathcal{\hat{P}}_{k}\odot\ldots\odot\mathcal{\hat{P}}_{2}\odot\mathcal{\hat{P}}_{1}\odot\mathcal{\hat{P}}_{i}$
in history Hilbert space $\breve{\mathcal{H}}$, there is a corresponding
chain operator $\hat{K}_{j}=\mathcal{\hat{P}}_{f}\hat{\mathcal{T}}_{f,k}\mathcal{\hat{P}}_{k}\hat{\mathcal{T}}_{k,k-1}\ldots\hat{\mathcal{T}}_{3,2}\mathcal{\hat{P}}_{2}\hat{\mathcal{T}}_{2,1}\mathcal{\hat{P}}_{1}\hat{\mathcal{T}}_{1,i}\mathcal{\hat{P}}_{i}$
in standard Hilbert space $\mathcal{H}$, where $\hat{\mathcal{T}}_{k,k-1}=\hat{\mathcal{T}}_{k-1,k}^{\dagger}$
is the time evolution operator from $t_{k-1}$ to $t_{k}$.
\end{defn}

\begin{defn}(History probability amplitude)
The quantum
probability amplitude $\psi_{j}$ propagating along a quantum history
$\hat{\mathcal{Q}}_{j}$ from an initial quantum state $|\psi_{i}\rangle$
at $t_{i}$ to a final quantum state $|\psi_{f}\rangle$ at $t_{f}$
is given by $\psi_{j}=\langle\psi_{f}|\hat{K}_{j}|\psi_{i}\rangle$.
%%%
Expanding the projectors using their corresponding unit eigenvectors
as $\hat{\mathcal{P}}_{k}=|\psi_{k}\rangle\langle\psi_{k}|$, allows
us to rewrite the chain operator $\hat{K}_{j}$ as $\hat{K}_{j}=|\psi_{f}\rangle\langle\psi_{f}|\hat{\mathcal{T}}_{f,k}|\psi_{k}\rangle\langle\psi_{k}|\hat{\mathcal{T}}_{k,k-1}|\psi_{k-1}\rangle\langle\psi_{k-1}|\times\ldots\hat{\mathcal{T}}_{3,2}|\psi_{2}\rangle\langle\psi_{2}|\hat{\mathcal{T}}_{2,1}|\psi_{1}\rangle\langle\psi_{1}|\hat{\mathcal{T}}_{1,i}|\psi_{i}\rangle\langle\psi_{i}|$.
Introducing the Feynman propagators from $|\psi_{k-1}\rangle$ to $|\psi_{k}\rangle$ as $F(\psi_{k}|\psi_{k-1})=\langle\psi_{k}|\hat{\mathcal{T}}_{k,k-1}|\psi_{k-1}\rangle$,
further gives $\hat{K}_{j}=|\psi_{f}\rangle F(\psi_{f}|\psi_{k})F(\psi_{k}|\psi_{k-1}) \ldots F(\psi_{2}|\psi_{1})F(\psi_{1}|\psi_{i})\langle\psi_{i}|$.
The quantum probability amplitude for the history is then a product
of Feynman propagators (each of which is a complex-valued function)
$\psi_{j}=\langle\psi_{f}|\hat{K}_{j}|\psi_{i}\rangle=F(\psi_{f}|\psi_{k})F(\psi_{k}|\psi_{k-1})\ldots F(\psi_{2}|\psi_{1})F(\psi_{1}|\psi_{i})$.
%%%
\end{defn}

\begin{defn}
\label{def:Feyn}(Feynman's sum-over-histories) The quantum probability
amplitude for a quantum transition from an initial quantum state $|\psi_{i}\rangle$
at $t_{i}$ to a final quantum state $|\psi_{f}\rangle$ at $t_{f}$
is given by the sum $\sum_{j}\psi_{j}$ over a complete set of orthogonal
quantum histories $\{\hat{\mathcal{Q}}_{1},\hat{\mathcal{Q}}_{2},\ldots,\hat{\mathcal{Q}}_{j}\}$,
$j\in\{1,2,\ldots,n^{k+2}\}$, which span the history Hilbert space
of the system $\breve{\mathcal{H}}=\mathcal{H}_{f}\odot\mathcal{H}_{k}\odot\ldots\odot\mathcal{H}_{2}\odot\mathcal{H}_{1}\odot\mathcal{H}_{i}$.
Inclusion of $\mathcal{\hat{P}}(\psi_{i})=|\psi_{i}\rangle\langle\psi_{i}|$ among the projectors of
the complete set at $t_{i}$ and $\mathcal{\hat{P}}(\psi_{f})=|\psi_{f}\rangle\langle\psi_{f}|$ among
the projectors of the complete set at $t_{f}$ eliminates a large
number of quantum histories $\hat{\mathcal{Q}}_{\perp}$ that start
or end with projection operators respectively orthogonal to $|\psi_{i}\rangle$
or $|\psi_{f}\rangle$, and consequently have zero contribution, $\langle\psi_{f}|\hat{K}_{\perp}|\psi_{i}\rangle=0$,
to the Feynman sum. Thus, Feynman summation will produce identical
result if it is performed over all orthogonal quantum histories of
the type $\mathcal{\hat{Q}}_{s}=\mathcal{\hat{P}}(\psi_{f})\odot\mathcal{\hat{P}}_{k}\odot\ldots\odot\mathcal{\hat{P}}_{2}\odot\mathcal{\hat{P}}_{1}\odot\mathcal{\hat{P}}(\psi_{i})$,
which form a complete set for the intermediate times $t_{1},t_{2},\ldots,t_{k}$.
The usage of the Feynman sum $\sum_{s}\mathcal{\hat{Q}}_{s}$, $s\in\{1,2,\ldots,n^{k}\}$
reduces the complete history Hilbert space for Feynman summation to
$n^{k}$-dimensional due to consideration of only the $k$ copies
of the $n$-dimensional Hilbert space at intermediate times $t_{1},t_{2},\ldots,t_{k}$.\end{defn}
\begin{thm}
\label{thm:continuous}Discontinuous Feynman histories have zero contribution to the total Feynman sum $\sum_{s}\psi_{s}$.
Feynman summation over a complete set of continuous quantum histories generates the same result as the total Feynman sum
$\sum_{s}\psi_{s}$ over all histories.
\end{thm}
\begin{proof}
The quantum probability amplitude $\psi_{s}$ propagating along an
arbitrary quantum history $\mathcal{\hat{Q}}_{s}=\mathcal{\hat{P}}(\psi_{f})\odot\mathcal{\hat{P}}_{k}\odot\ldots\odot\mathcal{\hat{P}}_{2}\odot\mathcal{\hat{P}}_{1}\odot\mathcal{\hat{P}}(\psi_{i})$,
is calculated from the inner product $\langle\psi_{f}|\hat{K}_{s}|\psi_{i}\rangle$
of the corresponding chain operator $\hat{K}_{s}=\mathcal{\hat{P}}(\psi_{f})\hat{\mathcal{T}}_{f,k}\mathcal{\hat{P}}_{k}\hat{\mathcal{T}}_{k,k-1}\ldots\hat{\mathcal{T}}_{3,2}\mathcal{\hat{P}}_{2}\hat{\mathcal{T}}_{2,1}\mathcal{\hat{P}}_{1}\hat{\mathcal{T}}_{1,i}\mathcal{\hat{P}}(\psi_{i})$.
The quantum time evolution operators $\hat{\mathcal{T}}_{k,k-1}$
are continuous in space and have non-zero product $\mathcal{\hat{P}}_{k}\hat{\mathcal{T}}_{k,k-1}\mathcal{\hat{P}}_{k-1}\neq0$
only between spatially connected projectors $\mathcal{\hat{P}}_{k-1}$ and $\mathcal{\hat{P}}_{k}$.
The presence of two consecutive disconnected projectors $\mathcal{\hat{P}}_{k-1}$
and $\mathcal{\hat{P}}_{k}$ anywhere in the quantum history effectively
zeroes it through the presence of $\mathcal{\hat{P}}_{k}\hat{\mathcal{T}}_{k,k-1}\mathcal{\hat{P}}_{k-1}=0$.\end{proof}

Next, let us briefly review the concept of weak values in Aharonov's
two-state vector formalism. Experimental measurement of weak values
requires a weak coupling between the measured system and the measuring
pointer, multiple experimental runs, post-selection and calculation
of averages \cite{Aharonov1988,Jozsa2007,Aharonov2014,Dressel2014}. Because unknown quantum states cannot be cloned \cite{Wootters1982},
weak values are meaningful only if one is given an ensemble $|\psi_{i}\rangle\otimes|\psi_{i}\rangle\otimes\ldots\otimes|\psi_{i}\rangle$
of quantum systems that are all prepared in the \emph{same }initial
quantum state $|\psi_{i}\rangle$ upon which measurements are made
and only those results are analyzed that end up with a certain post-selected
final state $|\psi_{f}\rangle$.
\begin{defn}
\label{defn:WV}
(Weak value) The \emph{weak value} of an operator $\hat{A}$ at any moment of
time $t_{m}$ during the evolution from initial state
$|\psi_{i}\rangle$ at an initial time $t_{i}$ to a final state $|\psi_{f}\rangle$
at a final time $t_{f}$ is
\begin{equation}
A_{w}=\frac{\langle\psi_{f}|\hat{\mathcal{T}}_{f,m}\hat{A}\hat{\mathcal{T}}_{m,i}|\psi_{i}\rangle}{\langle\psi_{f}|\hat{\mathcal{T}}_{f,i}|\psi_{i}\rangle}\label{eq:Aw}
\end{equation}
where $\hat{\mathcal{T}}_{f,i}=\hat{\mathcal{T}}_{f,m}\hat{\mathcal{T}}_{m,i}$ and $\langle\psi_{f}|\hat{\mathcal{T}}_{f,i}|\psi_{i}\rangle\neq0$.
In Aharonov's two-state vector formalism the pre-selected
state $|\psi_{i}\rangle$ evolves forward in time with the time evolution
operator $\hat{\mathcal{T}}_{m,i}$ and the post-selected state $|\psi_{f}\rangle$
evolves backward in time with the time evolution operator $\hat{\mathcal{T}}_{f,m}^{\dagger}=\hat{\mathcal{T}}_{m,f}$,
namely $(\hat{\mathcal{T}}_{f,m}^{\dagger}|\psi_{f}\rangle)^{\dagger}=\langle\psi_{f}|\hat{\mathcal{T}}_{f,m}$,
so that one employs both a bra and a ket at the same time $t_{m}$ at which $\hat{A}$
is measured \cite{Aharonov2017}.
\end{defn}

\begin{defn}
\label{defn:SWV}
(Sequential weak value) The concept of weak value can be generalized
into \emph{sequential weak value} of several operators $\hat{A}_{1}$,
$\hat{A}_{2}$, $\ldots$, $\hat{A}_{k}$ at several times $t_{1},t_{2},\ldots,t_{k}$
\cite{Mitchison2007,Diosi2016} as
\begin{equation}
(A_{k},\ldots,A_{2},A_{1})_{w}=\frac{\langle\psi_{f}|\hat{\mathcal{T}}_{f,k}\hat{A}_{k}\hat{\mathcal{T}}_{k,k-1}\ldots\hat{A}_{2}\hat{\mathcal{T}}_{2,1}\hat{A}_{1}\hat{\mathcal{T}}_{1,i}|\psi_{i}\rangle}{\langle\psi_{f}|\hat{\mathcal{T}}_{f,i}|\psi_{i}\rangle}\label{eq:Aw-seq}
\end{equation}
where $\hat{\mathcal{T}}_{f,i}=\hat{\mathcal{T}}_{f,k}\hat{\mathcal{T}}_{k,k-1}\ldots\hat{\mathcal{T}}_{2,1}\hat{\mathcal{T}}_{1,i}$.
\end{defn}
Weak values are complex-valued, however, both the real and the imaginary
parts of the weak values defined by Eqs.~\ref{eq:Aw} and \ref{eq:Aw-seq}
can be experimentally measured with the use of weak measurements (cf.
\cite{Aharonov1990,Jozsa2007,Mitchison2007,Svensson2013,Piacentini2016}).

%%%
The mathematical expressions \eqref{eq:Aw} and \eqref{eq:Aw-seq} of weak values arise in the approximate calculation of the pointer shifts when multiplying truncated power series expansions of the exponentiated interaction Hamiltonians between the measured system and the measuring pointers at $k$-times (see \hyperref[sec:Appendix]{Appendix}).

Before we present the main results of this work, we wish to address two technical points. First, we note that all the above was defined for arbitrary operators, but in the next section we shall focus on (not necessarily commuting) projection operators as commonly done when discussing sum over histories. Second, for making the notion of weak measurement feasible, the physical systems in question are assumed to exist in a fine-grained Hilbert space,
which can be taken to be either finite dimensional and consisting
of Planck scale units, or infinitely dimensional (so that standard
differential and integral calculus applies) but effectively described by a finite, coarse-grained Hilbert space. We shall henceforth assume an $n$-dimensional Hilbert space, applicable to the two cases above.
%%%

\section{Main results}
Now we are ready to demonstrate the tight relationship between Aharonov's two-state vector formalism and Feynman's sum-over-histories. We will also elucidate the meaning and properties of sequential weak values of multi-time projection operators.
\begin{thm}
\label{thm:main}The sequential weak value $(\mathcal{P}_{k},\ldots,\mathcal{P}_{2},\mathcal{P}_{1})_{w}$
of multi-time projection operators $\mathcal{\hat{P}}_{1},\mathcal{\hat{P}}_{2},\ldots,\mathcal{\hat{P}}_{k}$
at times $t_{1},t_{2},\ldots,t_{k}$
is equal to the quantum probability amplitude $\psi_{s'}$ propagating along the
individual Feynman history $\mathcal{\hat{Q}}_{s'}=\mathcal{\hat{P}}(\psi_{f})\odot\mathcal{\hat{P}}_{k}\odot\ldots\odot\mathcal{\hat{P}}_{2}\odot\mathcal{\hat{P}}_{1}\odot\mathcal{\hat{P}}(\psi_{i})$,
divided by the total quantum probability amplitude $\sum_{s}\psi_{s}$
of the Feynman sum $\sum_{s}\mathcal{\hat{Q}}_{s}=\mathcal{\hat{P}}(\psi_{f})\odot\hat{I}\odot\ldots\odot\hat{I}\odot\hat{I}\odot\mathcal{\hat{P}}(\psi_{i})=\sum_{s}\mathcal{\hat{P}}(\psi_{f})\odot\mathcal{\hat{P}}_{k,s}\odot\ldots\odot\mathcal{\hat{P}}_{2,s}\odot\mathcal{\hat{P}}_{1,s}\odot\mathcal{\hat{P}}(\psi_{i})$,
$s\in\{1,2,\ldots,n^{k}\}$ over a complete set of quantum histories
from~$|\psi_{i}\rangle$ to~$|\psi_{f}\rangle$.\end{thm}
\begin{proof}
The quantum probability amplitude for the individual Feynman history
$\mathcal{\hat{Q}}_{s'}=\mathcal{\hat{P}}(\psi_{f})\odot\mathcal{\hat{P}}_{k}\odot\ldots\odot\mathcal{\hat{P}}_{2}\odot\mathcal{\hat{P}}_{1}\odot\mathcal{\hat{P}}(\psi_{i})$
is given by the corresponding chain operator
\begin{eqnarray}
\psi_{s'} & = & \langle\psi_{f}|\hat{K}_{s'}|\psi_{i}\rangle\nonumber \\
 & = & \langle\psi_{f}|\mathcal{\hat{P}}(\psi_{f})\hat{\mathcal{T}}_{f,k}\mathcal{\hat{P}}_{k}\hat{\mathcal{T}}_{k,k-1}\ldots\mathcal{\hat{P}}_{2}\hat{\mathcal{T}}_{2,1}\mathcal{\hat{P}}_{1}\hat{\mathcal{T}}_{1,i}\mathcal{\hat{P}}(\psi_{i})|\psi_{i}\rangle\nonumber \\
 & = & \langle\psi_{f}|\hat{\mathcal{T}}_{f,k}\mathcal{\hat{P}}_{k}\hat{\mathcal{T}}_{k,k-1}\ldots\mathcal{\hat{P}}_{2}\hat{\mathcal{T}}_{2,1}\mathcal{\hat{P}}_{1}\hat{\mathcal{T}}_{1,i}|\psi_{i}\rangle\label{eq:two-vec-k}
\end{eqnarray}
which is exactly the numerator in Eq.~\ref{eq:Aw-seq}. Thus, two-state
vectors of multi-time projection operators in Aharonov's two-state
vector formalism are equivalent to quantum probability amplitudes
propagating along a Feynman history. Similarly, the total quantum
probability amplitude for the Feynman sum over all quantum histories
$\sum_{s}\mathcal{\hat{Q}}_{s}=\mathcal{\hat{P}}(\psi_{f})\odot\hat{I}\odot\ldots\odot\hat{I}\odot\hat{I}\odot\mathcal{\hat{P}}(\psi_{i})$
is given by the sum of all chain operators
\begin{eqnarray}
\sum_{s}\psi_{s} & = & \langle\psi_{f}|\sum_{s}\hat{K}_{s}|\psi_{i}\rangle\nonumber \\
 & = & \langle\psi_{f}|\mathcal{\hat{P}}(\psi_{f})\hat{\mathcal{T}}_{f,k}\hat{I}\hat{\mathcal{T}}_{k,k-1}\ldots\hat{I}\hat{\mathcal{T}}_{2,1}\hat{I}\hat{\mathcal{T}}_{1,i}\mathcal{\hat{P}}(\psi_{i})|\psi_{i}\rangle\nonumber \\
 & = & \langle\psi_{f}|\hat{\mathcal{T}}_{f,k}\hat{I}\hat{\mathcal{T}}_{k,k-1}\ldots\hat{I}\hat{\mathcal{T}}_{2,1}\hat{I}\hat{\mathcal{T}}_{1,i}|\psi_{i}\rangle\nonumber \\
 & = & \langle\psi_{f}|\hat{\mathcal{T}}_{f,i}|\psi_{i}\rangle\label{eq:two-vec-k-sum}
\end{eqnarray}
which is exactly the denominator in Eq.~\ref{eq:Aw-seq}. Eq.~\ref{eq:two-vec-k-sum}
also shows that the denominator of weak values in Aharonov's two-vector
state formalism is a disguised two-state vector of multi-time identity
operator. Dividing Eq.~\ref{eq:two-vec-k} by \ref{eq:two-vec-k-sum}
gives
\begin{equation}
(\mathcal{P}_{k},\ldots,\mathcal{P}_{2},\mathcal{P}_{1})_{w}=\frac{\psi_{s'}}{\sum_{s}\psi_{s}}\label{eq:main}
\end{equation}
Because the ordinary weak values (Eq.~\ref{eq:Aw}) serve as a special
single-time case of sequential weak values (Eq.~\ref{eq:Aw-seq}),
Eq.~\ref{eq:main} holds true for Feynman histories with a single
intermediate time as well. Interestingly, Eq.~\ref{eq:main} even
makes sense for the trivial case with no intermediate time points
in the quantum history $\mathcal{Q}=\mathcal{\hat{P}}(\psi_{f})\odot\mathcal{\hat{P}}(\psi_{i})$
where it returns the weak value of the identity operator $I_{w}=1$.
\end{proof}
%%%
Equation \ref{eq:main} provides a direct link between the theory of weak values in weak measurements, which require
a small, but strictly non-zero perturbation, i.e. $g>0$, and Feynman sum-over-histories, which exactly quantifies
quantum interference of virtual quantum histories without any external
coupling, i.e. $g=0$. Thus, we demonstrate unambiguously that weak
values are not an artifact arising from the small perturbation parameter $g$,
but are rather descriptive properties of quantum systems that are
exactly defined at $g=0$. For example, in experimental measurement of a single-time weak value, the pointer shift is $g\textrm{Re}\left[A_{w}\right]$
or $g\textrm{Im}\left[A_{w}\right]$ plus a higher order correction
term $\mathcal{O}(g^{3})$ (see \hyperref[sec:Appendix]{Appendix}), hence due to the pointer shift dependence
on $g$, the weak value can be measured with arbitrarily small, but
non-zero error $\mathcal{O}(g^{3})$. By considering the theory of
weak measurement alone, where weak values correspond to, and are interpreted as, average pointer shifts \cite{Wu2013,Shomroni2013}, one may be misled into thinking that the weak value
is only defined as a limit at $g\to 0$, while at $g=0$ due to the zero
pointer shift there is no weak value to be extracted. The mathematical
technique for Feynman summation, may however provide a proper context for better understanding
the meaning of weak values as relative quantum probability amplitudes
at zero disturbance. To measure such amplitudes, which by definition
are at zero disturbance ($g=0$), Aharonov \emph{et al.} \cite{Aharonov1988} developed
the weak measurement scheme that allows for controlling
the error in the measurement of the weak values, making the error
arbitrarily small for sufficiently small $g$.
%%%

From the measurability of weak values, we can prove that quantum probability
amplitudes along individual virtual Feynman histories entering into
a quantum superposed Feynman sum are also measurable given an ensemble
$|\psi_{i}\rangle\otimes|\psi_{i}\rangle\otimes\ldots\otimes|\psi_{i}\rangle$
of quantum systems that are all prepared in the \emph{same }initial
state $|\psi_{i}\rangle$.
\begin{thm}
Measured sequential weak value $(\mathcal{P}_{k},\ldots,\mathcal{P}_{2},\mathcal{P}_{1})_{w}$
of multi-time projection operators $\mathcal{\hat{P}}_{1},\mathcal{\hat{P}}_{2},\ldots,\mathcal{\hat{P}}_{k}$
could be converted (up to a pure phase factor $e^{\imath\theta}$)
into quantum amplitude $\psi_{s'}$ for the individual quantum history
$\mathcal{\hat{Q}}_{s'}=\mathcal{\hat{P}}(\psi_{f})\odot\mathcal{\hat{P}}_{k}\odot\ldots\odot\mathcal{\hat{P}}_{2}\odot\mathcal{\hat{P}}_{1}\odot\mathcal{\hat{P}}(\psi_{i})$
entering into a quantum superposed Feynman sum $\sum_{s}\psi_{s}$
via multiplication of the weak value $(\mathcal{P}_{k},\ldots,\mathcal{P}_{2},\mathcal{P}_{1})_{w}$
by the positive square root $|\langle\psi_{f}|\hat{\mathcal{T}}_{f,i}|\psi_{i}\rangle|$
of the experimentally measured quantum probability $p=|\langle\psi_{f}|\hat{\mathcal{T}}_{f,i}|\psi_{i}\rangle|^{2}$
for an initial pre-selected state $|\psi_{i}\rangle$ to end at the
final post-selected state $|\psi_{f}\rangle$.\end{thm}
\begin{proof}
From Eqs.~\ref{eq:Aw-seq} and \ref{eq:two-vec-k} we can express
$\psi_{s'}$ through the weak value as
\begin{equation}
\psi_{s'}=(\mathcal{P}_{k},\ldots,\mathcal{P}_{2},\mathcal{P}_{1})_{w}\langle\psi_{f}|\hat{\mathcal{T}}_{f,i}|\psi_{i}\rangle
\end{equation}
Since $\langle\psi_{f}|\hat{\mathcal{T}}_{f,i}|\psi_{i}\rangle$ is
a complex number it can be expressed as a product of its real-valued
modulus $\left|\langle\psi_{f}|\hat{\mathcal{T}}_{f,i}|\psi_{i}\rangle\right|$
times a pure phase $e^{\imath\theta}$. Thus, for the quantum probability
amplitude, we have
\begin{equation}
\psi_{s'}=(\mathcal{P}_{k},\ldots,\mathcal{P}_{2},\mathcal{P}_{1})_{w}|\langle\psi_{f}|\hat{\mathcal{T}}_{f,i}|\psi_{i}\rangle|e^{\imath\theta}\label{eq:psi-k}
\end{equation}
In the weak value formula (Eq.~\ref{eq:main}), the pure phase $e^{\imath\theta}$
is canceled down from the numerator and denominator. Because removing the pure
phase $e^{\imath\theta}$ from each of the superposed quantum histories
$\psi_{s}$ entering into the Feynman sum
\begin{equation}
\sum_{s}\psi_{s}=|\langle\psi_{f}|\hat{\mathcal{T}}_{f,i}|\psi_{i}\rangle|e^{\imath\theta}\label{eq:psi-k-sum}
\end{equation}
does not affect the quantum interference effects, the weak values
can be used to directly probe Feynman's sum-over-histories formulation
of quantum mechanics.
\end{proof}
%%%
Sequential weak values are defined with the use of quantum
observables $\hat{A}_{1}$, $\hat{A}_{2}$, $\ldots$,
$\hat{A}_{k}$ at $k$ times (Definition \ref{defn:SWV}). Therefore, in general,
sequential weak values are not the normalized quantum probability amplitudes propagating
along quantum histories. The spectral decompositions of
observables in Eq.~\ref{eq:Aw-seq} are given by
$\hat{A}_{1}=\sum_{n_{1}}\lambda_{n_{1}}\hat{\mathcal{P}}_{n_{1}}$,
$\hat{A}_{2}=\sum_{n_{2}}\lambda_{n_{2}}\hat{\mathcal{P}}_{n_{2}}$, $\ldots$,
$\hat{A}_{k}=\sum_{n_{k}}\lambda_{n_{k}}\hat{\mathcal{P}}_{n_{k}}$,
where $n_1$, $n_2$, $\ldots$, $n_k$ are indices that may vary independently, $\{\lambda_{n_{1}}\}$, $\{\lambda_{n_{2}}\}$, $\dots$, $\{\lambda_{n_{k}}\}$ are sets of eigenvalues and
$\{\hat{\mathcal{P}}_{n_{1}}\}$, $\{\hat{\mathcal{P}}_{n_{2}}\}$, $\ldots$, $\{\hat{\mathcal{P}}_{n_{k}}\}$ are sets of corresponding projection operators
for the eigenvectors of $\hat{A}_{1}$, $\hat{A}_{2}$, $\ldots$ $\hat{A}_{k}$. Consequently, a
general sequential weak value will be a weighted sum of quantum probability amplitudes
for Feynman histories, each of which is multiplied by a non-normalized
weight given by a product of eigenvalues $\lambda_{n_{1}}\lambda_{n_{2}}\ldots\lambda_{n_{k}}$.
To illustrate the point, let us set $\hat{H}=0$ to suppress all time
evolution operators i.e. $\hat{\mathcal{T}}_{k,k-1}=\hat{I}$, thereby
obtaining for the sequential weak value:
\begin{alignat}{1}
 & (A_{k},\ldots,A_{2},A_{1})_{w}=\frac{\langle\psi_{f}|\hat{A}_{k}\ldots\hat{A_{2}}\hat{A}_{1}|\psi_{i}\rangle}{\langle\psi_{f}|\psi_{i}\rangle}\nonumber \\
 & =\frac{\langle\psi_{f}|{\displaystyle \sum_{n_{k}}}\lambda_{n_{k}}\hat{\mathcal{P}}_{n_{k}}\ldots{\displaystyle \sum_{n_{2}}}\lambda_{n_{2}}\hat{\mathcal{P}}_{n_{2}}{\displaystyle \sum_{n_{1}}}\lambda_{n_{1}}\hat{\mathcal{P}}_{n_{1}}|\psi_{i}\rangle}{\langle\psi_{f}|\psi_{i}\rangle}\nonumber \\
 & ={\displaystyle \sum_{n_{1},n_{2},\ldots,n_{k}}}\lambda_{n_{1}}\lambda_{n_{2}}\ldots\lambda_{n_{k}}\frac{\langle\psi_{f}|\hat{\mathcal{P}}_{n_{k}}\ldots\hat{\mathcal{P}}_{n_{2}}\hat{\mathcal{P}}_{n_{1}}|\psi_{i}\rangle}{\langle\psi_{f}|\psi_{i}\rangle}\nonumber \\
\label{eq:gSWV}
\end{alignat}
Such a general sequential weak value $(A_{k},\ldots,A_{2},A_{1})_{w}$ is not subject
to the Born rule and does not generate a probability for observing
the corresponding quantum (Feynman) history (defined with the projectors
only). Our main point is that by restricting the general observables down to
\emph{projection operators} in sequential weak values, one can connect Feynman
sum-over-histories approach with the fruitful area of weak measurements
and weak values. Note that for each sequential weak value of multi-time projection operators,
$(\mathcal{P}_{k},\ldots,\mathcal{P}_{2},\mathcal{P}_{1})_{w}$ there
is a corresponding Feynman history and the probability for measuring
that history through a series of strong measurements at $k$ times is given by the Born rule,
i.e. $\textrm{Prob}\left[\mathcal{P}_{k},\ldots,\mathcal{P}_{2},\mathcal{P}_{1}\right]=\left|(\mathcal{P}_{k},\ldots,\mathcal{P}_{2},\mathcal{P}_{1})_{w}\langle\psi_{f}|\psi_{i}\rangle\right|^{2}$.
%%%

Sequential weak values of multi-time projection operators
are able to directly probe the quantum
probability amplitudes $\psi_{s}$ along individual virtual Feynman
histories that enter into a quantum superposed Feynman sum $\sum_{s}\psi_{s}$.
Because Feynman's sum-over-histories approach to quantum mechanics
works for a complete set of orthogonal quantum histories
in the history Hilbert space, we can derive an exact value for the
sum of the corresponding sequential weak values:
\begin{thm}
\label{thm:Weak-sum-1}For a complete set of orthogonal quantum histories
$\{\mathcal{\hat{Q}}_{1},\mathcal{\hat{Q}}_{2},\ldots,\mathcal{\hat{Q}}_{s}\}$
that span the history Hilbert space of a quantum transition with non-zero
probability, the complex sequential weak values sum up to unity $\sum_{s}(\mathcal{P}_{k,s},\ldots,\mathcal{P}_{2,s},\mathcal{P}_{1,s})_{w}=1$.\end{thm}
\begin{proof}
Quantum transition with non-zero probability ensures that all weak
values are finite due to non-zero denominator, $\sum_{s}\psi_{s}>0$.
Taking the sum over all histories $s\in\{1,2,\ldots,n^{k}\}$ on both sides of Eq.~\ref{eq:main}
gives
\begin{eqnarray}
\sum_{s}(\mathcal{P}_{k,s},\ldots,\mathcal{P}_{2,s},\mathcal{P}_{1,s})_{w} & = & \frac{\psi_{1}}{\sum_{s}\psi_{s}}+\frac{\psi_{2}}{\sum_{s}\psi_{s}}+\ldots+\frac{\psi_{s'}}{\sum_{s}\psi_{s}}\nonumber \\
 & = & \frac{\sum_{s}\psi_{s}}{\sum_{s}\psi_{s}}=1
\end{eqnarray}
\end{proof}
The converse of Theorem \ref{thm:Weak-sum-1} is not true, namely,
the fact that the weak values for a set of quantum histories
sum to unity $\sum_{s}(\mathcal{P}_{k,s},\ldots,\mathcal{P}_{2,s},\mathcal{P}_{1,s})_{w}=1$
does not imply that the set of quantum histories is complete.
\begin{cor}
Sequential weak values of multi-time projection operators are not
conditional probabilities, but relative probability amplitudes $(\mathcal{P}_{k},\ldots,\mathcal{P}_{2},\mathcal{P}_{1})_{w}=\frac{\psi_{s'}}{\sum_{s}\psi_{s}}$.
Weak values are measured by the mean value of the pointer shift of the measuring
device, which makes quantum probability amplitudes measurable provided
that one is given an ensemble $|\psi_{i}\rangle\otimes|\psi_{i}\rangle\otimes\ldots\otimes|\psi_{i}\rangle$
of quantum systems that are all prepared in the \emph{same }initial
state $|\psi_{i}\rangle$.\end{cor}
\begin{thm}
\label{thm:Feyn-Schrod}Analysis of quantum interference effects within a complete set of mutually orthogonal quantum histories $\{\mathcal{\hat{Q}}_{1},\hat{\mathcal{Q}}_{2},\ldots,\hat{\mathcal{Q}}_{s}\}$ from $|\psi_{i}\rangle$ to $|\psi_{f}\rangle$ is consistent with the standard quantum
mechanical picture.\end{thm}
\begin{proof}
By the completeness of the set of quantum histories
entering into the Feynman sum, we are guaranteed to obtain identity
operators for all intermediate times $\sum_{s}\hat{\mathcal{Q}}_{s}=\hat{\mathcal{P}}(\psi_{f})\odot\hat{I}\odot\ldots\odot\hat{I}\odot\mathcal{\hat{P}}(\psi_{i})$.
Therefore, the corresponding sum of chain operators is
\begin{eqnarray}
\sum_{s}\hat{K}_{s} & = & \hat{\mathcal{P}}(\psi_{f})\hat{\mathcal{T}}_{f,k}\hat{I}\hat{\mathcal{T}}_{k,k-1}\ldots\hat{\mathcal{T}}_{2,1}\hat{I}\hat{\mathcal{T}}_{1,i}\mathcal{\hat{P}}(\psi_{i})\nonumber \\
 & = & \mathcal{\hat{P}}(\psi_{f})\hat{\mathcal{T}}_{f,i}\mathcal{\hat{P}}(\psi_{i})\label{eq:sum-1}
\end{eqnarray}
Expressing $\hat{\mathcal{T}}_{f,i}$ in terms of the Hamiltonian shows that the total Feynman sum is just the standard quantum probability amplitude that one would obtain from
unitary evolution according to the Schr\"{o}dinger equation
\begin{equation}
\langle\psi_{f}|\sum_{s}\hat{K}_{s}|\psi_{i}\rangle=\langle\psi_{f}|\hat{\mathcal{T}}_{f,i}|\psi_{i}\rangle=\langle\psi_{f}|
e^{-\frac{\imath}{\hbar}\int_{t_{i}}^{t_{f}}\hat{H}(t)\,dt}|\psi_{i}\rangle\label{eq:Schrod}
\end{equation}
Noteworthy, orthogonality of the corresponding chain operators $\{\hat{K}_{1},\hat{K}_{2},\ldots,\hat{K}_{s}\}$ was not assumed, which shows that Feynman summation is not equivalent to the decoherent (consistent) histories approach that requires $\textrm{Tr}\left(\hat{K}_{j}\hat{K}_{j'}\right)=0$ for $j\neq j'$ \cite{Gell-Mann1990,Gell-Mann1993,Griffiths1984,Griffiths1993,Griffiths2003,Hartle1993,Halliwell1995}.
\end{proof}
Analysis of weak values corresponding to a complete set
of mutually orthogonal quantum histories that span the history Hilbert
space avoids paradoxes because the orthogonality ensures that one weak
value cannot be used to infer claims for more than one history,
and the completeness of the set of histories implies consistency with
the Schr\"{o}dinger equation (Theorem \ref{thm:Feyn-Schrod}). Due to
the linearity of sums in quantum mechanical inner products $\langle\psi_{f}|\sum_{s}\hat{K}_{s}|\psi_{i}\rangle=\sum_{s}\langle\psi_{f}|\hat{K}_{s}|\psi_{i}\rangle$,
Feynman's approach provides a natural language for discussion of quantum
interference effects between individual quantum histories \cite{Cotler2015,Cotler2016,Cotler2017,Nowakowski2017}.
%%%
Running the proof of Theorem \ref{thm:Feyn-Schrod} backwards also shows that
starting from the Schr\"{o}dinger equation (Eq.~\ref{eq:Schrod}), one
could obtain correct quantum probability amplitudes
by inserting identity operators
at intermediate time points and then summing over all quantum histories
spanning the history Hilbert space (Eq.~\ref{eq:sum-1}).
%%%
\begin{thm}
\label{thm:seq-w-superposition}Sequential weak values evaluated at different
number of intermediate times correspond to different coarse-grainings
of the history Hilbert space. Consequently, $(k-1)$-time sequential
weak values are quantum superpositions of $k$-time sequential weak
values.\end{thm}
\begin{proof}
The sequential weak value $(\hat{\mathcal{P}}_{k-1},\ldots,\hat{\mathcal{P}}_{2},\hat{\mathcal{P}}_{1})_{w}$
corresponds to a $(k-1)$-time coarse-grained Feynman history $\hat{\mathcal{P}}(\psi_{f})\odot\hat{\mathcal{P}}_{k-1}\odot\ldots\odot\hat{\mathcal{P}}_{2}\odot\hat{\mathcal{P}}_{1}\odot\mathcal{\hat{P}}(\psi_{i}).$
The time tensor $\odot$ between projectors at $t_{f}$ and $t_{k-1}$
contains a hidden identity operator $\hat{I}_{k}$ at time $t_{k}$,
which when resolved as a sum of orthogonal projectors gives a quantum
superposition of $k$-time fine-grained Feynman histories $\hat{\mathcal{P}}(\psi_{f})\odot\hat{\mathcal{P}}_{k-1}\odot\ldots\odot\hat{\mathcal{P}}_{2}\odot\hat{\mathcal{P}}_{1}\odot\mathcal{\hat{P}}(\psi_{i})=\hat{\mathcal{P}}(\psi_{f})\odot\hat{I}_{k}\odot\hat{\mathcal{P}}_{k-1}\odot\ldots\odot\hat{\mathcal{P}}_{2}\odot\hat{\mathcal{P}}_{1}\odot\mathcal{\hat{P}}(\psi_{i})=\sum_{n}\hat{\mathcal{P}}(\psi_{f})\odot\hat{\mathcal{P}}_{(k,n)}\odot\hat{\mathcal{P}}_{k-1}\odot\ldots\odot\hat{\mathcal{P}}_{2}\odot\hat{\mathcal{P}}_{1}\odot\mathcal{\hat{P}}(\psi_{i})$.
Calculating the quantum probability amplitudes from the corresponding
chain operators and applying the weak value formula (Eq.~\ref{eq:main})
gives
\begin{eqnarray}
(\hat{\mathcal{P}}_{k-1},\ldots,\hat{\mathcal{P}}_{2},\hat{\mathcal{P}}_{1})_{w} & = & \sum_{n}(\hat{\mathcal{P}}_{(k,n)},\hat{\mathcal{P}}_{k-1},\ldots,\hat{\mathcal{P}}_{2},\hat{\mathcal{P}}_{1})_{w}\nonumber \\
 & = & (\hat{I}_{k},\hat{\mathcal{P}}_{k-1},\ldots,\hat{\mathcal{P}}_{2},\hat{\mathcal{P}}_{1})_{w}
\end{eqnarray}
\end{proof}
%%%
\begin{defn}
\emph{Incompatible weak values} are (sequential) weak values whose corresponding
quantum histories are not orthogonal in history Hilbert space.
\end{defn}
The main goal of Feynman sum-over-histories is to predict
probabilities for quantum events to occur. To obtain valid
quantum probabilities, however, the Feynman summation should not be performed
over all quantum histories in the history Hilbert space $\breve{\mathcal{H}}$,
but only over a complete set of orthogonal histories that
span $\breve{\mathcal{H}}$.
The non-orthogonal quantum histories
of incompatible weak values cannot interfere
%(either coherently or incoherently)
with each
other because this would overcount certain histories
in the Feynman sum more than once, rendering incorrect quantum probability
for the transition from $|\psi_{i}\rangle$ to $|\psi_{f}\rangle$ for almost all physically valid Hamiltonians.
Indeed, consider a complete set of quantum histories $\sum_s \hat{Q}_s=\mathcal{\hat{P}}(\psi_{f})\odot\hat{I}\odot\mathcal{\hat{P}}(\psi_{i})$ to which is added an extra non-orthogonal history $\hat{Q}_{s'}$. In the general case with $\psi_{s'}\neq0$, for coherent superposition, we will have
\begin{equation}
\left|\frac{\psi_{s'} + \sum_s \psi_{s}}{\sum_s \psi_{s}}\right|^2= \left|\frac{\psi_{s'}}{\sum_s \psi_{s}} + 1\right|^2\neq 1
\end{equation}
and for incoherent superposition
\begin{equation}
\frac{|\psi_{s'}|^2 + \sum_s |\psi_{s}|^2}{\sum_s |\psi_{s}|^2}= \frac{|\psi_{s'}|^2}{\sum_s |\psi_{s}|^2} + 1\neq 1
\end{equation}
Conservation of quantum probability will not be violated only in the special case where
%the extra non-orthogonal history contributes a zero quantum probability amplitude,
$\psi_{s'}=0$. Thus, one may be tempted to give a special status to non-orthogonal quantum histories with zero weak values and interpret them unconditionally. This, however, would contradict the mathematical principles that ensure the status of Feynman sum-over-histories as one of several equivalent formulations of quantum mechanics. In particular, notice that the orthogonality of quantum histories is independent of the Hamiltonian $\hat{H}$ and the correctly constructed Feynman sum $\sum_s \hat{Q}_s=\mathcal{\hat{P}}(\psi_{f})\odot\hat{I}\odot\mathcal{\hat{P}}(\psi_{i})$ will always return the correct transition amplitude $\sum_s \psi_s$ for any $\hat{H}$. On the other hand, having a zero quantum probability amplitude, $\psi_{s'}=0$, is a Hamiltonian-dependent condition, which means that
%an arbitrarily small change in the Hamiltonian leading to $\psi_{s'}\neq0$ (e.g. due to incomplete destructive interference) will immediately reveal the inconsistency of summing over non-orthogonal histories through non-conservation of probability. In other words,
summation over non-orthogonal histories cannot return the correct transition amplitudes for all physically valid Hamiltonians, hence it cannot be a fundamental principle upon which to build quantum mechanics.
%%%
\section{Application}
We illustrate the power of the presented theorems with the
analysis of a concrete interferometric setup shown in Fig.~\ref{fig:1}.
%%%
The transition probability amplitude $\psi_{(S\to D)}$ from the source $S$ to the detector $D$ can be easily calculated with the use of actual Feynman summation and various weak values can be determined with the use of Theorem \ref{thm:main}. Among the three alternative ways to calculate the Feynman sum, namely with the use of matrix exponential of the Hamiltonian $\hat{H}(t)$, time development operators $\hat{\mathcal{T}}_{k,k-1}$
 or Feynman propagators $F(\psi_{k}|\psi_{k-1})$, the latter one is computationally most effective.
%%%
Utilizing Theorem~\ref{thm:continuous}, there are only nine
coarse-grained continuous quantum histories $\hat{\mathcal{Q}}_{s}$ from
$S$ to $D$ that need to be summed over with their corresponding quantum probability amplitudes $\psi_{s}$:
\begin{eqnarray*}
\hat{\mathcal{Q}}_{1} & = & \hat{D}\odot\hat{x}_{1}\odot\hat{x}_{1}\odot\hat{x}_{1}\odot\hat{x}_{1}\odot\hat{x}_{1}\odot\hat{S},\quad\psi_{1}=+2^{-\frac{B_{n}+2}{2}}\\
\hat{\mathcal{Q}}_{2} & = & \hat{D}\odot\hat{x}_{9}\odot\hat{x}_{7}\odot\hat{x}_{5}\odot\hat{x}_{3}\odot\hat{x}_{2}\odot\hat{S},\quad\psi_{2}=+2^{-{3}}\\
\hat{\mathcal{Q}}_{3} & = & \hat{D}\odot\hat{x}_{9}\odot\hat{x}_{8}\odot\hat{x}_{5}\odot\hat{x}_{3}\odot\hat{x}_{2}\odot\hat{S},\quad\psi_{3}=-2^{-{3}}\\
\hat{\mathcal{Q}}_{4} & = & \hat{D}\odot\hat{x}_{9}\odot\hat{x}_{7}\odot\hat{x}_{6}\odot\hat{x}_{3}\odot\hat{x}_{2}\odot\hat{S},\quad\psi_{4}=-2^{-{3}}\\
\hat{\mathcal{Q}}_{5} & = & \hat{D}\odot\hat{x}_{9}\odot\hat{x}_{8}\odot\hat{x}_{6}\odot\hat{x}_{3}\odot\hat{x}_{2}\odot\hat{S},\quad\psi_{5}=-2^{-{3}}\\
\hat{\mathcal{Q}}_{6} & = & \hat{D}\odot\hat{x}_{9}\odot\hat{x}_{7}\odot\hat{x}_{5}\odot\hat{x}_{4}\odot\hat{x}_{2}\odot\hat{S},\quad\psi_{6}=+2^{-{3}}\\
\hat{\mathcal{Q}}_{7} & = & \hat{D}\odot\hat{x}_{9}\odot\hat{x}_{8}\odot\hat{x}_{5}\odot\hat{x}_{4}\odot\hat{x}_{2}\odot\hat{S},\quad\psi_{7}=-2^{-{3}}\\
\hat{\mathcal{Q}}_{8} & = & \hat{D}\odot\hat{x}_{9}\odot\hat{x}_{7}\odot\hat{x}_{6}\odot\hat{x}_{4}\odot\hat{x}_{2}\odot\hat{S},\quad\psi_{8}=+2^{-{3}}\\
\hat{\mathcal{Q}}_{9} & = & \hat{D}\odot\hat{x}_{9}\odot\hat{x}_{8}\odot\hat{x}_{6}\odot\hat{x}_{4}\odot\hat{x}_{2}\odot\hat{S},\quad\psi_{9}=+2^{-{3}}
\end{eqnarray*}
\begin{figure}[t]
\begin{centering}
\includegraphics[width=85mm]{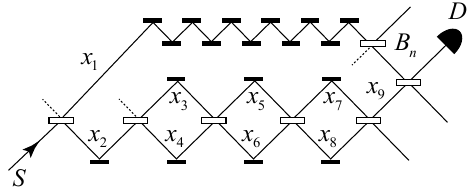}
\par\end{centering}
\caption{\label{fig:1}An interferometer with nine coarse-grained continuous quantum histories from the source $S$ to the detector $D$. The weak value $(x_{i})_{w}$ of the position projector $\hat{x}_{i}=|x_{i}\rangle\langle x_{i}|$ is the relative quantum probability amplitude of the sum of all histories that pass through $x_{i}$ divided by the sum of all histories from $S$ to $D$. $B_{n}$ is a variable number of beam splitters that can be used to reduce the quantum probability amplitude that reaches the detector~$D$ along history $x_{1}$.
%%%
Dashed lines indicate inactive sources of quanta that are required for the construction of the Hamiltonian $\hat{H}(t)$; solid lines indicate paths constructed as products of Feynman propagators.
%%%
}
\end{figure}
For $B_{n}=4$, the single-time weak value $(x_{1})_{w}=+1$ is able
to extract the quantum probability amplitude $\psi_{1}$ along
history $\hat{\mathcal{Q}}_{1}$, however, for
histories $\hat{\mathcal{Q}}_{2}$--$\hat{\mathcal{Q}}_{9}$ one needs
to use sequential weak values of multi-time projection operators that
uniquely identify each history inside the three inner interferometers:
$(x_{7},x_{5},x_{3})_{w}=+1$, $(x_{8},x_{5},x_{3})_{w}=-1$, $(x_{7},x_{6},x_{3})_{w}=-1$,
$(x_{8},x_{6},x_{3})_{w}=-1$, $(x_{7},x_{5},x_{4})_{w}=+1$, $(x_{8},x_{5},x_{4})_{w}=-1$,
$(x_{7},x_{6},x_{4})_{w}=+1$, $(x_{8},x_{6},x_{4})_{w}=+1$.

From Eq.~\ref{eq:main} it can be seen that once the fine-grained
quantum histories are resolved, adding projectors at extra times does
not change the weak values, e.g. $(x_{7},x_{5},x_{3})_{w}=(x_{7},x_{5},x_{3},x_{2})_{w}=(x_{9},x_{7},x_{5},x_{3})_{w}=(x_{9},x_{7},x_{5},x_{3},x_{2})_{w}=+1$.
On the other hand, reducing the number of projectors selects quantum superpositions of Feynman histories, e.g.:
\begin{eqnarray*}
(x_{2})_{w} & = & \frac{\psi_{2}+\psi_{3}+\psi_{4}+\psi_{5}+\psi_{6}+\psi_{7}+\psi_{8}+\psi_{9}}{\sum_{s}\psi_{s}}=0\\
(x_{3})_{w} & = & \frac{\psi_{2}+\psi_{3}+\psi_{4}+\psi_{5}}{\sum_{s}\psi_{s}}=-2\\
(x_{4})_{w} & = & \frac{\psi_{6}+\psi_{7}+\psi_{8}+\psi_{9}}{\sum_{s}\psi_{s}}=+2
\end{eqnarray*}
Thus, weak values are descriptive properties of the measured quantum
system that depend on the quantum history of interest (Theorem~\ref{thm:main}).
Feynman's sum-over-histories emphasizes the natural occurrence of pre- and post-selection in quantum mechanics. Moreover, it also reveals that in some sense sequential weak values are primitive and more fundamental than single-time weak values, which
are in fact superposed sums of sequential weak values, e.g.: $(x_{3})_{w}=(x_{7},x_{5},x_{3})_{w}+(x_{8},x_{5},x_{3})_{w}+(x_{7},x_{6},x_{3})_{w}+(x_{8},x_{6},x_{3})_{w}=+1-1-1-1=-2$.
This was similarly shown for multipartite weak values \cite{Aharonov2017b}.

Weak values measure different Feynman histories from the source $S$
to the detector $D$, but only sets of weak values that complete the
history Hilbert space can be consistently interpreted together. For
example, taken together $(x_{1})_{w}=+1$ and $(x_{2})_{w}=0$ state
that the quantum has reached the detector $D$ through $x_{1}$ but
not through $x_{2}$, and this is consistent because all fine-grained
histories $\hat{\mathcal{Q}}_{1}$--$\hat{\mathcal{Q}}_{9}$ are accounted
for. In contrast, when taken together $(x_{1})_{w}=+1$, $(x_{2})_{w}=0$,
$(x_{3})_{w}=-2$, $(x_{4})_{w}=+2$, $(x_{5})_{w}=0$, $(x_{6})_{w}=0$,
$(x_{7})_{w}=+2$ , $(x_{8})_{w}=-2$ and $(x_{9})_{w}=0$ state that
the quantum has not passed through $x_{2}$ and $x_{9}$, yet it has
been at $x_{3}$, $x_{4}$, $x_{7}$ and $x_{8}$; the apparent discontinuity
arises from overcounting five times each of $\psi_{2}$--$\psi_{9}$.
Thus, Theorem \ref{thm:seq-w-superposition} explicitly addresses
the controversy between Svensson and Vaidman \cite{Svensson2015,BenIsrael2017,Svensson2017} utilizing the general applicability of weak values for determining the history of a quantum system.

While weak values substantiate the physical nature of virtual Feynman
histories through measurable pointer shifts, the mathematical constraints
for correct Feynman summation elucidate the meaning and properties
of weak values. Sequential weak values reflect the unique character of temporal correlations, as was also shown by Avella \emph{et al.} \cite{Avella2017}. Consider as another example, an experimenter changing the
number of beam splitters from $B_{n}=4$ to $B_{n}=22$ on the history through $x_{1}$,
while measuring devices record the weak values at $x_{3}$ or $x_{4}$.
The presence of $18$ extra beamsplitters on arm $x_{1}$ is
felt by the weak measuring devices at arm $x_{3}$ or $x_{4}$ as
they measure the very large weak values $(x_{3})_{w}=-1024$ and
$(x_{4})_{w}=+1024$. In other words, the weak measurement devices at arms $x_{3}$ or $x_{4}$
somehow feel the photon exploration of alternative quantum histories \cite{Danan2013}.
Thus, the weak value measured through some weak coupling to a measuring pointer at one location integrates information about the presence of other devices at different locations in the interferometer through the change of the total Feynman sum $\sum_{s}\psi_{s}$.
Of course, weak values cannot be used for superluminal communication since to extract the weak values
from the recorded data, experimenters located at $x_{3}$ or $x_{4}$
need to know which photons were detected by~$D$.
\section{Concluding remarks}
Our results are consistent with a recent work by Sokolovski \cite{Sokolovski2016b},
but we have extended it in scope and generality. First, we have shown
that weak values should be interpreted for complete sets of quantum
histories, because they provide information for the phase difference
between any two histories in the complete set. Second, our Theorem~\ref{thm:main}
is completely general and gives the quantum probability amplitude
along any quantum history in terms of a corresponding sequential weak
value (Eq.~\ref{eq:main}), which reduces to a single-time weak value
in the special case of a history with a single intermediate time.
Third, in regard to the measurability of virtual Feynman histories,
our work builds upon previous results on measurability
of weak values \cite{Jozsa2007,Mitchison2007,Shikano2012,Svensson2013,Lu2014}. For a single-time weak value, the mean value of the pointer shift in the measuring device is proportional to the weak coupling factor $g\ll 1$
\cite{Jozsa2007,Shikano2012,Svensson2013,Lu2014}. For a multi-time sequential weak
value at $k$ times, the mean value of the pointer shift is
proportional to $g^{k}$ \cite{Mitchison2007}, which makes it equally
harder to measure the quantum probability amplitudes for the corresponding
multi-time Feynman histories. Furthermore, to evaluate the expectation value of a product of $N$ pointer positions, one needs in general not just the $N$-point sequential weak value, but also all other $n$-point ones, for $1 \le n \le N-1$. 
%%%
However, there is a clear way in principle for measuring sequential weak values: Initially, the measured projectors have to be weakly coupled to a set of ancillary pointers and then the correlation between pointers' states has to be projectively measured (see \hyperref[sec:Appendix]{Appendix}). This has been experimentally demonstrated in \cite{Piacentini2016}, where for each photon the sequential weak value of two projections on incompatible polarization states were measured through weak coupling to the transverse displacements. This method is also of practical importance, allowing to perform quantum state tomography \cite{Thekkadath2016} and quantum process tomography \cite{Ber2013}.

%%%
To conclude, we have presented and analyzed
the tight relation between Feynman's sum-over-histories and sequential
weak values and shown how one formalism corroborates the other, proving some new theorems.
%%%
This analysis may strengthen the fundamental role previously ascribed to weak values \cite{Vaidman1996,Vaidman2017,Dressel2014,Dressel2015,Williams2008,Pusey2014} and at the same time might make Feynman's histories more tangible, amenable to direct experimental observation.
%%%

\section*{Acknowledgements}
We wish to thank Yakir Aharonov, Bengt Svensson and Dmitri Sokolovski for helpful comments and discussions. We also thank Alexandre Matzkin, who suggested us some interesting literature to examine, and three anonymous referees for very helpful comments. E.C. was supported by the Canada Research Chairs (CRC) Program.

%%%
\section{\label{sec:Appendix}Appendix}
For making the paper self-contained, we outline below the theory of single-time and sequential weak values. These results are mostly known in literature, but they are vital for understanding our claims above and especially how sequential weak values can be measured in practice.

\subsection{\label{sec:MWVs}Measurement of single-time weak values}

For simplicity, the measuring device $M$ starts with a real-valued Gaussian position
wave function centered at zero
\begin{equation}
\phi(x)=\left(2\pi\sigma^{2}\right)^{-\frac{1}{4}}e^{-\frac{x^{2}}{4\sigma^{2}}}\label{eq:meter-0}
\end{equation}
which gives a corresponding Gaussian distribution
\begin{equation}
\phi^{2}(x)=\frac{1}{\sqrt{2\pi\sigma^{2}}}e^{-\frac{(x-\mu)^{2}}{2\sigma^{2}}}
\end{equation}
with position mean $\mu_{0}(x)=0$ and variance $\sigma^2_{0}(x)=\sigma^2$.

The interaction Hamiltonian between the measured system $S$ and the
measuring device $M$ is
\begin{equation}
\hat{H}_{\textrm{int}}=g\delta(t-t_{m})\,\hat{A}\otimes\hat{p}
\end{equation}
where $\hat{A}$ is an observable for the measured system $S$ and
$\hat{p}=\hbar\hat{k}=-\imath\hbar\frac{\partial}{\partial x}$ is
the meter variable conjugate to the meter pointer variable $\hat{x}$.
Allowing the measured system $S$ to evolve with internal Hamiltonian
$\hat{H}_{S}\otimes\hat{I}$ and suppressing the internal Hamiltonian
of the meter $\hat{I}\otimes\hat{H}_{M}=0$, we obtain for the composite
time evolution operator
\begin{eqnarray}
\hat{\mathcal{T}}_{\textrm{composite}} & = & e^{-\frac{\imath}{\hbar}\int_{t_{i}}^{t_{f}}\left[\hat{H}_{S}\otimes\hat{I}+g\delta(t-t_{m})\hat{A}\otimes\hat{p}\right]dt}\nonumber \\
 & = & e^{-\frac{\imath}{\hbar}\hat{H}_{S}\otimes\hat{I}\,(t_{f}-t_{m})}e^{-\frac{\imath}{\hbar}g\,\hat{A}\otimes\hat{p}}e^{-\frac{\imath}{\hbar}\hat{H}_{S}\otimes\hat{I}\,(t_{m}-t_{i})}\nonumber \\
 & = & \hat{\mathcal{T}}_{f,m}\,e^{-\frac{\imath}{\hbar}g\,\hat{A}\otimes\hat{p}}\hat{\mathcal{T}}_{m,i}\label{eq:evol-1}
\end{eqnarray}
Hereafter, we will use $\hat{\mathcal{T}}_{k,k-1}=e^{-\frac{\imath}{\hbar}\hat{H}_{S}\otimes\hat{I}\,(t_{k}-t_{k-1})}$
to compress the internal time evolution operators of the measured
system $S$.

\subsubsection{Real part of weak value}

The composite system starts from initial state
\begin{equation}
|\psi_{i}\rangle|\phi\rangle=|\psi_{i}\rangle\int_{-\infty}^{\infty}\phi(x)|x\rangle dx
\end{equation}
and evolves with the time evolution operator in Eq. \ref{eq:evol-1}.
Due to small $g$ satisfying $g\ll1$, we can use a truncated power
series at $\mathcal{O}(g^{3})$ for the interaction term. For post-selected
system in a final state $|\psi_{f}\rangle$, the final meter wave
function in position basis is
\begin{eqnarray}
 & & \langle x|\phi_{f}\rangle=\langle\psi_{f}|\hat{\mathcal{T}}_{\textrm{composite}}|\psi_{i}\rangle\phi(x)\nonumber \\
 & & \approx\langle\psi_{f}|\hat{\mathcal{T}}_{f,m}\left(1-\imath g\hat{A}\otimes\hat{k}-\frac{g^{2}}{2}\hat{A}^{2}\otimes\hat{k}^{2}\right)\hat{\mathcal{T}}_{m,i}|\psi_{i}\rangle\phi(x)\nonumber \\
\end{eqnarray}
where we used $\langle x|x^{\prime}\rangle=\delta(x-x^{\prime})$
and the integral property of Dirac's delta function $\int_{-\infty}^{\infty}\phi(x^{\prime})\delta(x-x^{\prime})dx^{\prime}=\phi(x)$.

Expressing the wave number operator in position basis $\hat{k}=-\imath\frac{\partial}{\partial x}$
and using Lagrange's notation for spatial partial derivatives gives
\begin{eqnarray}
\phi_{f} & \approx & \langle\psi_{f}|\psi_{i}\rangle\left(\phi-gA_{w}\phi^{\prime}+\frac{g^{2}}{2}(A^{2})_{w}\phi^{\prime\prime}\right)
\end{eqnarray}
The normalized final meter distribution is
\begin{eqnarray}
\frac{|\phi_{f}|^{2}}{|\langle\psi_{f}|\psi_{i}\rangle|^{2}} & \approx & \left|\phi-gA_{w}\phi^{\prime}+\frac{g^{2}}{2}(A^{2})_{w}\phi^{\prime\prime}\right|^{2}\nonumber \\
 & \approx & \phi^{2}-g\left[\overline{A_{w}}+A_{w}\right]\phi\phi^{\prime}+g^{2}A_{w}\overline{A_{w}}\phi^{\prime}\phi^{\prime}\nonumber \\
 & & +\frac{g^{2}}{2}\left[\overline{(A^{2})_{w}}+(A^{2})_{w}\right]\phi\phi^{\prime\prime}\label{eq:final-1}
\end{eqnarray}
The mean (expected value of position) of the normalized final meter
distribution $|\phi_{f}|^{2}$ is calculated as the first raw moment
\begin{equation}
\langle x\rangle=\frac{\int_{-\infty}^{\infty}x\,|\phi_{f}|^{2}dx}{|\langle\psi_{f}|\psi_{i}\rangle|^{2}}\label{eq:ave}
\end{equation}
Taking into account the exact initial meter wavefunction in Eq. \ref{eq:meter-0},
which is real and centered at zero, we have
\begin{eqnarray}
\int_{-\infty}^{\infty}x\,\left(\phi^{2}\right)dx & = & 0\label{eq:raw-x-1}\\
\int_{-\infty}^{\infty}x\,\left(\phi\phi^{\prime}\right)dx & = & -\frac{1}{2}\label{eq:raw-x-2}\\
\int_{-\infty}^{\infty}x\,\left(\phi^{\prime}\phi^{\prime}\right)dx & = & 0\label{eq:raw-x-3}\\
\int_{-\infty}^{\infty}x\,\left(\phi\phi^{\prime\prime}\right)dx & = & 0\label{eq:raw-x-4}
\end{eqnarray}
With the above equations, from Eqs. \ref{eq:final-1} and \ref{eq:ave},
we get
\begin{eqnarray}
\langle x\rangle & = & \frac{g}{2}\left[\overline{A_{w}}+A_{w}\right]=g\,\textrm{Re}(A_{w})
\end{eqnarray}
So the mean value of final meter distribution in position basis $x$
measures the real part of the weak value $A_{w}$.

\subsubsection{Imaginary part of weak value}

Fourier transform of the initial meter position quantum wave function
to wave number basis $|k\rangle$ gives
\begin{equation}
|\psi_{i}\rangle|\phi\rangle=|\psi_{i}\rangle\int_{-\infty}^{\infty}\tilde{\phi}(k)|k\rangle dk
\end{equation}
Using a truncated power series at $\mathcal{O}(g^{3})$ for the interaction
term, we obtain for the final meter state
\begin{eqnarray}
 & & \langle k|\phi_{f}\rangle=\langle\psi_{f}|\hat{\mathcal{T}}_{\textrm{composite}}|\psi_{i}\rangle\tilde{\phi}(k)\nonumber \\
 & & \approx\langle\psi_{f}|\hat{\mathcal{T}}_{f,m}\left(1-\imath g\,\hat{A}\otimes\hat{k}-\frac{g^{2}}{2}\hat{A}^{2}\otimes\hat{k}^{2}\right)\hat{\mathcal{T}}_{m,i}|\psi_{i}\rangle\tilde{\phi}(k)\nonumber \\
\end{eqnarray}
and the final meter wave function in wave number basis

\begin{equation}
\tilde{\phi}_{f}\approx\langle\psi_{f}|\psi_{i}\rangle\tilde{\phi}\Bigg[1-\imath gA_{w}k-\frac{1}{2}g^{2}(A^{2})_{w}k^{2}\Bigg]
\end{equation}
where the initial Gaussian wave number wave function is real and centered
at zero
\begin{equation}
\tilde{\phi}(k)=\left(\frac{2\sigma^{2}}{\pi}\right)^{\frac{1}{4}}e^{-k^{2}\sigma^{2}}
\end{equation}
and the corresponding initial wave number probability distribution
\begin{equation}
\tilde{\phi}^{2}(k)=\sqrt{\frac{2}{\pi}}\sigma e^{-2k^{2}\sigma^{2}}
\end{equation}
is centered at $\mu_{0}(k)=0$ and has a variance $\sigma^2_{0}(k)=\frac{1}{4\sigma^{2}}$.

The normalized final meter distribution is
\begin{eqnarray}
\frac{|\tilde{\phi}_{f}|^{2}}{|\langle\psi_{f}|\psi_{i}\rangle|^{2}} & \approx & \left|\tilde{\phi}\Bigg[1-\imath gA_{w}k-\frac{g^{2}}{2}(A^{2})_{w}k^{2}\Bigg]\right|^{2}\nonumber \\
 & \approx & \tilde{\phi}^{2}\Bigg[1+\imath g\Big[\overline{A_{w}}-A_{w}\Big]k\nonumber \\
 & & -\frac{g^{2}}{2}\Big[\overline{(A^{2})_{w}}+(A^{2})_{w}-2A_{w}\overline{A_{w}}\Big]k^{2}\Bigg]\nonumber \\
\label{eq:final-2}
\end{eqnarray}
With the use of the following identities
\begin{eqnarray}
\int_{-\infty}^{\infty}k\,\left(\tilde{\phi}^{2}\right)dk & = & 0\label{eq:raw-k-1}\\
\int_{-\infty}^{\infty}k\,\left(k\tilde{\phi}^{2}\right)dk & = & \frac{1}{4\sigma^{2}}\label{eq:raw-k-2}\\
\int_{-\infty}^{\infty}k\,\left(k^{2}\tilde{\phi}^{2}\right)dk & = & 0\label{eq:raw-k-3}
\end{eqnarray}
from Eq. \ref{eq:final-2} we obtain that the mean value of the wave
number probability distribution is shifted from zero to
\begin{eqnarray}
\langle k\rangle & = & \frac{\imath g\left[\overline{A_{w}}-A_{w}\right]}{4\sigma^{2}}=\frac{g\,\textrm{Im}(A_{w})}{2\sigma^{2}}
\end{eqnarray}

\subsection{\label{sec:MSWVs}Measurement of two-time sequential weak values}

Consider two meter probes measuring two different observables $\hat{A}_{1}$
and $\hat{A}_{2}$ at two different times $t_{1}$ and $t_{2}$. The
interaction Hamiltonian between the measured system $S$ and the measuring
devices $M_{1}$ and $M_{2}$ is
\begin{equation}
\hat{H}_{\textrm{int}}=g\delta(t-t_{2})\,\hat{A}_{2}\otimes\hat{p}_{2}+g\delta(t-t_{1})\,\hat{A}_{1}\otimes\hat{p}_{1}
\end{equation}

The time evolution operator is
\begin{eqnarray}
\hat{\mathcal{T}}_{\textrm{composite}} & = & e^{-\frac{\imath}{\hbar}\int_{t_{i}}^{t_{f}}\left[\hat{H}_{S}\otimes\hat{I}+g\delta(t-t_{2})\,\hat{A}_{2}\otimes\hat{p}_{2}+g\delta(t-t_{1})\,\hat{A}_{1}\otimes\hat{p}_{1}\right]dt}\nonumber \\
 & = & \hat{\mathcal{T}}_{f,2}\,e^{-\frac{\imath}{\hbar}g\,\hat{A}_{2}\otimes\hat{p}_{2}}\hat{\mathcal{T}}_{2,1}e^{-\frac{\imath}{\hbar}g\,\hat{A}_{1}\otimes\hat{p}_{1}}\hat{\mathcal{T}}_{1,i}\label{eq:evol-2}
\end{eqnarray}

\subsubsection{Real part of sequential weak value}

\emph{Product $\langle x_{1}x_{2}\rangle$.}
Measuring both meter probes in $x$-basis $\langle x_{1}x_{2}\rangle$
extracts the real part of the sequential weak value plus an extra
term.

The composite system starts from initial state
\begin{equation}
|\psi_{i}\rangle|\phi_{1}\rangle|\phi_{2}\rangle=|\psi_{i}\rangle\int_{-\infty}^{\infty}\phi(x_{1})|x_{1}\rangle dx\int_{-\infty}^{\infty}\phi(x_{2})|x_{2}\rangle dx_{2}
\end{equation}
and evolves with the time evolution operator in Eq. \ref{eq:evol-2}.
Due to small $g$ satisfying $g\ll1$, we can use a truncated power
series at $\mathcal{O}(g^{3})$ for the interaction term. For post-selected
system in a final state $|\psi_{f}\rangle$, the final two-meter wave
function in position basis, $\hat{k}_{1}=-\imath\frac{\partial}{\partial x_{1}}$,
$\hat{k}_{2}=-\imath\frac{\partial}{\partial x_{2}}$, is
\begin{eqnarray}
\langle x_{1}|\langle x_{2}|\Phi_{f}\rangle & = & \langle\psi_{f}|\hat{\mathcal{T}}_{\textrm{composite}}|\psi_{i}\rangle\phi(x_{1})\phi(x_{2})\nonumber \\
 & \approx & \langle\psi_{f}|\hat{\mathcal{T}}_{f,2}\left(1-g\hat{A}_{2}\frac{\partial}{\partial x_{2}}+\frac{g^{2}}{2}\hat{A}_{2}^{2}\frac{\partial^{2}}{\partial x_{2}^{2}}\right)\nonumber \\
 & & \times\hat{\mathcal{T}}_{2,1}\left(1-g\hat{A}_{1}\frac{\partial}{\partial x_{1}}+\frac{g^{2}}{2}\hat{A}_{1}^{2}\frac{\partial^{2}}{\partial x_{1}^{2}}\right)\hat{\mathcal{T}}_{1,i}\nonumber \\
 & & \times\,|\psi_{i}\rangle\phi(x_{1})\phi(x_{2})
\end{eqnarray}
Multiplying the brackets and discarding $\mathcal{O}(g^{3})$ terms
gives
\begin{eqnarray}
\Phi_{f} & \approx & \langle\psi_{f}|\psi_{i}\rangle\Big[\phi_{1}\phi_{2}-g(A_{1})_{w}\phi_{1}^{\prime}\phi_{2}\nonumber \\
 & & -g(A_{2})_{w}\phi_{1}\phi_{2}^{\prime}+\frac{g^{2}}{2}(A_{1}^{2})_{w}\phi_{1}^{\prime\prime}\phi_{2}\nonumber \\
 & & +\frac{g^{2}}{2}(A_{2}^{2})_{w}\phi_{1}\phi_{2}^{\prime\prime}+g^{2}(A_{2},A_{1})_{w}\phi_{1}^{\prime}\phi_{2}^{\prime}\Big]
\end{eqnarray}
The normalized final meter distribution is
\begin{eqnarray}
\frac{|\Phi_{f}|^{2}}{|\langle\psi_{f}|\psi_{i}\rangle|^{2}} & \approx & \phi_{1}^{2}\phi_{2}^{2}-g\left[\overline{(A_{1})_{w}}+(A_{1})_{w}\right]\phi_{1}\phi_{1}^{\prime}\phi_{2}^{2}\nonumber \\
 & & -g\left[\overline{(A_{2})_{w}}+(A_{2})_{w}\right]\phi_{1}^{2}\phi_{2}\phi_{2}^{\prime}\nonumber \\
 & & +\frac{g^{2}}{2}\left[\overline{(A_{1}^{2})_{w}}+(A_{1}^{2})_{w}\right]\phi_{1}\phi_{1}^{\prime\prime}\phi_{2}^{2}\nonumber \\
 & & +\frac{g^{2}}{2}\left[\overline{(A_{2}^{2})_{w}}+(A_{2}^{2})_{w}\right]\phi_{1}^{2}\phi_{2}\phi_{2}^{\prime\prime}\nonumber \\
 & & +g^{2}\left[\overline{(A_{2},A_{1})_{w}}+(A_{2},A_{1})_{w}\right]\phi_{1}\phi_{1}^{\prime}\phi_{2}\phi_{2}^{\prime}\nonumber \\
 & & +g^{2}\left[\overline{(A_{1})_{w}}(A_{2})_{w}+(A_{1})_{w}\overline{(A_{2})_{w}}\right]\phi_{1}\phi_{1}^{\prime}\phi_{2}\phi_{2}^{\prime}\nonumber \\
 & & +g^{2}(A_{1})_{w}\overline{(A_{1})_{w}}\phi_{1}^{\prime}\phi_{1}^{\prime}\phi_{2}^{2}\nonumber \\
 & & +g^{2}(A_{2})_{w}\overline{(A_{2})_{w}}\phi_{1}^{2}\phi_{2}^{\prime}\phi_{2}^{\prime}
\end{eqnarray}
With the use of the identities (\ref{eq:raw-x-1}--\ref{eq:raw-x-4}),
we get
\begin{eqnarray}
\langle x_{1}x_{2}\rangle & = & \int_{-\infty}^{\infty}\int_{-\infty}^{\infty}x_{1}x_{2}\left(\frac{|\Phi_{f}|^{2}}{|\langle\psi_{f}|\psi_{i}\rangle|^{2}}\right)dx_{1}dx_{2}\nonumber \\
 & = & \frac{g^{2}}{4}\left[\overline{(A_{2},A_{1})_{w}}+(A_{2},A_{1})_{w}\right]\nonumber \\
 & & +\frac{g^{2}}{4}\left[\overline{(A_{1})_{w}}(A_{2})_{w}+(A_{1})_{w}\overline{(A_{2})_{w}}\right]\nonumber \\
 & = & \frac{g^{2}}{2}\textrm{Re}\left[(A_{2},A_{1})_{w}+\overline{(A_{1})_{w}}(A_{2})_{w}\right]\label{eq:swv-xx}
\end{eqnarray}

\emph{Product $\langle k_{1}k_{2}\rangle$.}
Measuring both meter probes in $k$-basis $\langle k_{1}k_{2}\rangle$
extracts the real part of the sequential weak value with a negative
sign plus an extra term.

The Fourier transform of the initial composite state is
\begin{equation}
|\psi_{i}\rangle|\phi_{1}\rangle|\phi_{2}\rangle=|\psi_{i}\rangle\int_{-\infty}^{\infty}\tilde{\phi}(k_{1})|k_{1}\rangle dk_{1}\int_{-\infty}^{\infty}\tilde{\phi}(k_{2})|k_{2}\rangle dk_{2}
\end{equation}
Using a truncated power series at $\mathcal{O}(g^{3})$ for the interaction
term, we obtain for the final meter state
\begin{eqnarray}
\langle k_{1}|\langle k_{2}|\tilde{\Phi}_{f}\rangle & = & \langle\psi_{f}|\hat{\mathcal{T}}_{\textrm{composite}}|\psi_{i}\rangle\tilde{\phi}(k_{1})\tilde{\phi}(k_{2})\nonumber \\
 & \approx & \langle\psi_{f}|\hat{\mathcal{T}}_{f,2}\left(1-\imath g\hat{A}_{2}\otimes\hat{k}_{2}-\frac{g^{2}}{2}\hat{A}_{2}^{2}\otimes\hat{k}_{2}^{2}\right)\nonumber \\
 & & \times\hat{\mathcal{T}}_{2,1}\left(1-\imath g\hat{A}_{1}\otimes\hat{k}_{1}-\frac{g^{2}}{2}\hat{A}_{1}^{2}\otimes\hat{k}_{1}^{2}\right)\hat{\mathcal{T}}_{1,i}\nonumber \\
 & & \times\,|\psi_{i}\rangle\tilde{\phi}(k_{1})\tilde{\phi}(k_{2})
\end{eqnarray}
which gives
\begin{eqnarray}
\tilde{\Phi}_{f} & \approx & \langle\psi_{f}|\psi_{i}\rangle\tilde{\phi}_{1}\tilde{\phi}_{2}\Bigg[1-\imath g(A_{1})_{w}k_{1}\nonumber \\
 & & -\imath g(A_{2})_{w}k_{2}-\frac{g^{2}}{2}(A_{1}^{2})_{w}k_{1}^{2}\nonumber \\
 & & -\frac{g^{2}}{2}(A_{2}^{2})_{w}k_{2}^{2}-g^{2}(A_{2},A_{1})_{w}k_{1}k_{2}\Bigg]\nonumber \\
\end{eqnarray}
The normalized final meter distribution is
\begin{eqnarray}
\frac{|\tilde{\Phi}_{f}|^{2}}{|\langle\psi_{f}|\psi_{i}\rangle|^{2}} & \approx & \tilde{\phi}_{1}^{2}\tilde{\phi}_{2}^{2}\Bigg[1+\imath g\left[\overline{(A_{1})_{w}}-(A_{1})_{w}\right]k_{1}\nonumber \\
 & & +\imath g\left[\overline{(A_{2})_{w}}-(A_{2})_{w}\right]k_{2}\nonumber \\
 & & -g^{2}\left[\overline{(A_{2},A_{1})_{w}}+(A_{2},A_{1})_{w}\right]k_{1}k_{2}\nonumber \\
 & & +g^{2}\left[\overline{(A_{1})_{w}}(A_{2})_{w}+(A_{1})_{w}\overline{(A_{2})_{w}}\right]k_{1}k_{2}\nonumber \\
 & & -\frac{g^{2}}{2}\left[\overline{(A_{1}^{2})_{w}}+(A_{1}^{2})_{w}-2(A_{1})_{w}\overline{(A_{1})_{w}}\right]k_{1}^{2}\nonumber \\
 & & -\frac{g^{2}}{2}\left[\overline{(A_{2}^{2})_{w}}+(A_{2}^{2})_{w}-2(A_{2})_{w}\overline{(A_{2})_{w}}\right]k_{2}^{2}\Bigg]\nonumber \\
\end{eqnarray}
With the use of the identities (\ref{eq:raw-k-1}--\ref{eq:raw-k-3}),
we get
\begin{eqnarray}
\langle k_{1}k_{2}\rangle & = & \int_{-\infty}^{\infty}\int_{-\infty}^{\infty}k_{1}k_{2}\left(\frac{|\tilde{\Phi}_{f}|^{2}}{|\langle\psi_{f}|\psi_{i}\rangle|^{2}}\right)dk_{1}dk_{2}\nonumber \\
 & = & -\frac{g^{2}}{16\sigma^{4}}\left[\overline{(A_{2},A_{1})_{w}}+(A_{2},A_{1})_{w}\right]\nonumber \\
 & & +\frac{g^{2}}{16\sigma^{4}}\left[\overline{(A_{1})_{w}}(A_{2})_{w}+(A_{1})_{w}\overline{(A_{2})_{w}}\right]\nonumber \\
 & = & \frac{g^{2}}{8\sigma^{4}}\textrm{Re}\left[-(A_{2},A_{1})_{w}+\overline{(A_{1})_{w}}(A_{2})_{w}\right]\label{eq:swv-kk}
\end{eqnarray}
Subtracting suitably scaled $\langle x_{1}x_{2}\rangle$ and $\langle k_{1}k_{2}\rangle$
gives only the real part of the two-time sequential weak value without
the extra terms of individual measurements
\begin{equation}
\textrm{Re}\left[(A_{2},A_{1})_{w}\right]=\frac{1}{g^{2}}\left(\langle x_{1}x_{2}\rangle-4\sigma^{4}\langle k_{1}k_{2}\rangle\right)\label{eq:SWV-Re}
\end{equation}

\subsubsection{Imaginary part of sequential weak value}

To extract the imaginary part of the sequential weak value, we need
to use mixed products. Again, there will be extra terms that need
to be subtracted.

\emph{Product $\langle x_{1}k_{2}\rangle$.}
To calculate $\langle x_{1}k_{2}\rangle$, we rewrite the initial
state of the composite system in a mixed product form
\begin{equation}
|\psi_{i}\rangle|\phi_{1}\rangle|\phi_{2}\rangle=|\psi_{i}\rangle\int_{-\infty}^{\infty}\phi(x_{1})|x_{1}\rangle dx\int_{-\infty}^{\infty}\tilde{\phi}(k_{2})|k_{2}\rangle dk_{2}
\end{equation}
Due to small $g$ satisfying $g\ll1$, we can use a truncated power
series at $\mathcal{O}(g^{3})$ for the interaction term. For post-selected
system in a final state $|\psi_{f}\rangle$, the final two-meter wave
function in position basis is
\begin{eqnarray}
\langle x_{1}|\langle k_{2}|\Phi_{f}\rangle & = & \langle\psi_{f}|\hat{\mathcal{T}}_{\textrm{composite}}|\psi_{i}\rangle\phi(x_{1})\tilde{\phi}(k_{2})\nonumber \\
 & \approx & \langle\psi_{f}|\hat{\mathcal{T}}_{f,2}\left(1-\imath g\hat{A}_{2}\otimes\hat{k}_{2}-\frac{g^{2}}{2}\hat{A}_{2}^{2}\otimes\hat{k}_{2}^{2}\right)\nonumber \\
 & & \times\hat{\mathcal{T}}_{2,1}\left(1-g\hat{A}_{1}\frac{\partial}{\partial x_{1}}+\frac{g^{2}}{2}\hat{A}_{1}^{2}\frac{\partial^{2}}{\partial x_{1}^{2}}\right)\hat{\mathcal{T}}_{1,i}\nonumber \\
 & & \times\,|\psi_{i}\rangle\phi(x_{1})\tilde{\phi}(k_{2})
\end{eqnarray}
Multiplying the brackets and discarding $\mathcal{O}(g^{3})$ terms
gives
\begin{eqnarray}
\Phi_{f} & \approx & \langle\psi_{f}|\psi_{i}\rangle\tilde{\phi}_{2}\Bigg[\phi_{1}-g(A_{1})_{w}\phi_{1}^{\prime}-\imath g(A_{2})_{w}k_{2}\phi_{1}\nonumber \\
 & & +\frac{g^{2}}{2}(A_{1}^{2})_{w}\phi_{1}^{\prime\prime}-\frac{g^{2}}{2}(A_{2}^{2})_{w}k_{2}^{2}\phi_{1}\nonumber \\
 & & +\imath g^{2}(A_{2},A_{1})_{w}k_{2}\phi_{1}^{\prime}\Bigg]
\end{eqnarray}
The normalized final meter distribution is
\begin{eqnarray}
\frac{|\Phi_{f}|^{2}}{|\langle\psi_{f}|\psi_{i}\rangle|^{2}} & \approx & \tilde{\phi}_{2}\Bigg[\phi_{1}^{2}-g\left[\overline{(A_{1})_{w}}+(A_{1})_{w}\right]\phi_{1}\phi_{1}^{\prime}\nonumber \\
 & & +\imath g\left[\overline{(A_{2})_{w}}-(A_{2})_{w}\right]k_{2}\phi_{1}^{2}\nonumber \\
 & & +g^{2}(A_{1})_{w}\overline{(A_{1})_{w}}\phi_{1}^{\prime}\phi_{1}^{\prime}+g^{2}(A_{2})_{w}\overline{(A_{2})_{w}}k_{2}^{2}\phi_{1}^{2}\nonumber \\
 & & -\imath g^{2}\left[\overline{(A_{2},A_{1})_{w}}-(A_{2},A_{1})_{w}\right]k_{2}\phi_{1}\phi_{1}^{\prime}\nonumber \\
 & & +\imath g^{2}\left[\overline{(A_{1})_{w}}(A_{2})_{w}-(A_{1})_{w}\overline{(A_{2})_{w}}\right]k_{2}\phi_{1}\phi_{1}^{\prime}\nonumber \\
 & & +\frac{g^{2}}{2}\left[\overline{(A_{1}^{2})_{w}}+(A_{1}^{2})_{w}\right]\phi_{1}\phi_{1}^{\prime\prime}\nonumber \\
 & & -\frac{g^{2}}{2}\left[(A_{2}^{2})_{w}+\overline{(A_{2}^{2})_{w}}\right]k_{2}^{2}\phi_{1}^{2}\Bigg]
\end{eqnarray}
With the use of the identities (\ref{eq:raw-x-1}--\ref{eq:raw-x-4}
and \ref{eq:raw-k-1}--\ref{eq:raw-k-3}), we get
\begin{eqnarray}
\langle x_{1}k_{2}\rangle & = & \int_{-\infty}^{\infty}\int_{-\infty}^{\infty}x_{1}k_{2}\left(\frac{|\Phi_{f}|^{2}}{|\langle\psi_{f}|\psi_{i}\rangle|^{2}}\right)dx_{1}dk_{2}\nonumber \\
 & = & \imath\frac{g^{2}}{8\sigma^{2}}\left[\overline{(A_{2},A_{1})_{w}}-(A_{2},A_{1})_{w}\right]\nonumber \\
 & & -\imath\frac{g^{2}}{8\sigma^{2}}\left[\overline{(A_{1})_{w}}(A_{2})_{w}-(A_{1})_{w}\overline{(A_{2})_{w}}\right]\nonumber \\
 & = & \frac{g^{2}}{4\sigma^{2}}\textrm{Im}\left[(A_{2},A_{1})_{w}-(A_{1})_{w}\overline{(A_{2})_{w}}\right]\label{eq:swv-xk}
\end{eqnarray}

\emph{Product $\langle k_{1}x_{2}\rangle$.}
To calculate $\langle k_{1}x_{2}\rangle$, we express the initial
composite state as the mixed product
\begin{equation}
|\psi_{i}\rangle|\phi_{1}\rangle|\phi_{2}\rangle=|\psi_{i}\rangle\int_{-\infty}^{\infty}\tilde{\phi}(k_{1})|k_{1}\rangle dk_{1}\int_{-\infty}^{\infty}\phi(x_{2})|x_{2}\rangle dx_{2}
\end{equation}
Using a truncated power series at $\mathcal{O}(g^{3})$ for the interaction
term, we obtain for the final meter state
\begin{eqnarray}
\langle k_{1}|\langle x_{2}|\Phi_{f}\rangle & = & \langle\psi_{f}|\hat{\mathcal{T}}_{\textrm{composite}}|\psi_{i}\rangle\tilde{\phi}(k_{1})\phi(x_{2})\nonumber \\
 & \approx & \langle\psi_{f}|\hat{\mathcal{T}}_{f,2}\left(1-g\hat{A}_{2}\frac{\partial}{\partial x_{2}}+\frac{g^{2}}{2}\hat{A}_{2}^{2}\frac{\partial^{2}}{\partial x_{2}^{2}}\right)\nonumber \\
 & & \times\hat{\mathcal{T}}_{2,1}\left(1-\imath g\hat{A}_{1}\otimes\hat{k}_{1}-\frac{g^{2}}{2}\hat{A}_{1}^{2}\otimes\hat{k}_{1}^{2}\right)\hat{\mathcal{T}}_{1,i}\nonumber \\
 & & \times\,|\psi_{i}\rangle\tilde{\phi}(k_{1})\phi(x_{2})
\end{eqnarray}
which gives
\begin{eqnarray}
\Phi_{f} & \approx & \langle\psi_{f}|\psi_{i}\rangle\tilde{\phi}_{1}\Bigg[\phi_{2}-\imath g(A_{1})_{w}k_{1}\phi_{2}-g(A_{2})_{w}\phi_{2}^{\prime}\nonumber \\
 & & -\frac{g^{2}}{2}(A_{1}^{2})_{w}k_{1}^{2}\phi_{2}+\frac{g^{2}}{2}(A_{2}^{2})_{w}\phi_{2}^{\prime\prime}\nonumber \\
 & & +\imath g^{2}(A_{2},A_{1})_{w}k_{1}\phi_{2}^{\prime}\Bigg]
\end{eqnarray}
The normalized final meter distribution is
\begin{eqnarray}
\frac{|\Phi_{f}|^{2}}{|\langle\psi_{f}|\psi_{i}\rangle|^{2}} & \approx & \tilde{\phi}_{1}^{2}\Bigg[\phi_{2}^{2}+\imath g\left[\overline{(A_{1})_{w}}-(A_{1})_{w}\right]k_{1}\phi_{2}^{2}\nonumber \\
 & & -g\left[\overline{(A_{2})_{w}}+(A_{2})_{w}\right]\phi_{2}\phi_{2}^{\prime}\nonumber \\
 & & +g^{2}(A_{1})_{w}\overline{(A_{1})_{w}}k_{1}^{2}\phi_{2}^{2}+g^{2}(A_{2})_{w}\overline{(A_{2})_{w}}\phi_{2}^{\prime}\phi_{2}^{\prime}\nonumber \\
 & & -\imath g^{2}\left[\overline{(A_{2},A_{1})_{w}}-(A_{2},A_{1})_{w}\right]k_{1}\phi_{2}\phi_{2}^{\prime}\nonumber \\
 & & -\imath g^{2}\left[\overline{(A_{1})_{w}}(A_{2})_{w}-(A_{1})_{w}\overline{(A_{2})_{w}}\right]k_{1}\phi_{2}\phi_{2}^{\prime}\nonumber \\
 & & -\frac{g^{2}}{2}\left[\overline{(A_{1}^{2})_{w}}+(A_{1}^{2})_{w}\right]k_{1}^{2}\phi_{2}^{2}\nonumber \\
 & & +\frac{g^{2}}{2}\left[\overline{(A_{2}^{2})_{w}}+(A_{2}^{2})_{w}\right]\phi_{2}\phi_{2}^{\prime\prime}
\end{eqnarray}
With the use of the identities (\ref{eq:raw-k-1}--\ref{eq:raw-k-3}),
we get
\begin{eqnarray}
\langle k_{1}x_{2}\rangle & = & \int_{-\infty}^{\infty}\int_{-\infty}^{\infty}k_{1}x_{2}\left(\frac{|\Phi_{f}|^{2}}{|\langle\psi_{f}|\psi_{i}\rangle|^{2}}\right)dk_{1}dx_{2}\nonumber \\
 & = & \imath\frac{g^{2}}{8\sigma^{2}}\left[\overline{(A_{2},A_{1})_{w}}-(A_{2},A_{1})_{w}\right]\nonumber \\
 & & +\imath\frac{g^{2}}{8\sigma^{2}}\left[\overline{(A_{1})_{w}}(A_{2})_{w}-(A_{1})_{w}\overline{(A_{2})_{w}}\right]\nonumber \\
 & = & \frac{g^{2}}{4\sigma^{2}}\textrm{Im}\left[(A_{2},A_{1})_{w}+(A_{1})_{w}\overline{(A_{2})_{w}}\right]\label{eq:swv-kx}
\end{eqnarray}
Adding suitably scaled $\langle x_{1}k_{2}\rangle$ and $\langle k_{1}x_{2}\rangle$
gives only the imaginary part of the two-time sequential weak value
without the extra terms of individual measurements
\begin{eqnarray}
\textrm{Im}\left[(A_{2},A_{1})_{w}\right] & = & \frac{2\sigma^{2}}{g^{2}}\left(\langle x_{1}k_{2}\rangle+\langle k_{1}x_{2}\rangle\right)\label{eq:SWV-Im}
\end{eqnarray}
In practice, experimental measurement of two-time sequential weak values does not use
Eqs. \ref{eq:SWV-Re} or \ref{eq:SWV-Im}, but rather directly subtracts
the product of single-time weak values using Eqs. \ref{eq:swv-xx}, \ref{eq:swv-kk},
\ref{eq:swv-xk}, or \ref{eq:swv-kx}, as in \cite{Piacentini2016}.
The reason is that it is easier to measure pointer shifts proportional
to $g$ instead of $g^{2}$.

%%%

\bibliographystyle{apsrev4-1}
\bibliography{references}

\end{document}